\documentclass[pdflatex,sn-mathphys-ay]{sn-jnl}

\usepackage{graphicx}%
\usepackage{multirow}%
\usepackage{amsmath,amssymb,amsfonts}%
\usepackage{amsthm}%
\usepackage{graphicx}
\usepackage{booktabs}
\usepackage{array}
\usepackage{adjustbox}
\usepackage{mathrsfs}%
\usepackage[title]{appendix}%
\usepackage{xcolor}%
\usepackage{textcomp}%
\usepackage{manyfoot}%
\usepackage{booktabs}%
\usepackage{algorithm}%
\usepackage{algorithmicx}%
\usepackage{algpseudocode}%
\usepackage{listings}%

\theoremstyle{thmstyleone}%
\newtheorem{theorem}{Theorem}
\newtheorem{proposition}[theorem]{Proposition}%

\theoremstyle{thmstyletwo}%

\theoremstyle{thmstylethree}%

\begin{document}

\title[Article Title]{Discrete $p$--Fold Directional Models for Binned Orientation Data with
Applications to Sand-Grain Orientations and Urban Street Networks}


\author[1]{\fnm{Brajesh Kumar} \sur{Dhakad}}\email{brajeshdhakad85047@gmail.com}

\author*[2]{\fnm{Jayant} \sur{Jha}}\email{jayantjha@gmail.com}

\affil[1]{\orgdiv{Department of Mathematics}, \orgname{Indian Institute of Technology Madras}, \orgaddress{\city{Chennai}, \postcode{600036}, \state{Tamil Nadu}, \country{India}}}

\affil*[2]{\orgdiv{Statistical Sciences Division}, \orgname{Indian Statistical Institute, Kolkata}, \orgaddress{\city{Kolkata}, \postcode{700108}, \state{West Bengal}, \country{India}}}


\abstract{
 Discrete directional data exhibiting $p$-fold rotational symmetry, where directions separated by integer multiples of $2\pi/p$ are equivalent, arise in various applications but are often modeled using continuous distributions. In this article, we develop two general methods for constructing distributions for such data: symmetrization and folding of discrete circular distributions, and wrapping of discrete linear distributions. Particular models arising from these constructions are studied, together with their properties and estimation methods. For axial sand-grain orientation data, the proposed models are compared with discretized continuous axial distributions, whose category probabilities require numerical integration. The proposed models attain smaller Akaike information criterion values among the models considered while allowing direct likelihood evaluation without numerical integration. The framework is also applied to binned Manhattan street-orientation data, demonstrating its applicability to higher-order rotational structures, specifically the $4$-fold case.}




\maketitle

\section{Introduction}\label{sec1}

Directional statistics is a specialized branch of statistics concerned primarily with observations represented as unit vectors on manifolds such as the circle, sphere, and torus. It has applications in several areas, including animal navigation, environmental science, computer vision, geology, and geography. The progress of this field is evident in several books published over the decades, including \cite{lewis1993, Fisher1993, jammalamadaka2001, mardia2009, ley2017, ley2018}.

In directional statistics, circular data consists of observations (directions in the plane) that can be uniquely represented by an angle $\theta \in [0, 2\pi)$ measured counterclockwise from a fixed reference axis. In some applications, different directions represent the same physical orientation. More precisely, for a given positive integer $p$, the direction $\theta$ and $\theta + 2\pi/p$ are regarded as equivalent. Such observations are also referred to as observations that exhibit $p$--fold rotational symmetry. The observations corresponding to $p =1$ are called circular, whereas those with $p =2$ are axial. Axial data consist of unoriented axes (or random lines) for which the two opposite directions cannot be distinguished. A higher value of $p$ arises when the directions separated by an angle of $2\pi/p$ represent the same orientation pattern. For example, an orthogonal street grid produces a $4$--fold pattern, for which directions separated by $90^\circ$ are treated as equivalent.

For continuous observations exhibiting such rotational symmetry, some modelling approaches have been developed. A classical approach for handling $p$--fold rotational symmetry is to multiply each observation by $p$, thereby mapping the observations into \([0,2\pi)\). This transformation is discussed in Chapter 2 of \cite{Fisher1993} and Chapter 3 of \cite{mardia2009}. Subsequent methodological developments have focused predominantly on the axial ($p = 2$) case. \cite{arnold2006} proposed probabilistic models for axial data, including the axial normal distribution, and \cite{abe2011} introduced cosine-perturbed circular distributions by doubling the argument of a symmetric unimodal density. Further developments include the sine-skewed distributions proposed by \cite{abe2012}. Other recent advances in continuous axial modelling include Bayesian mixture frameworks \citep{fan2021} and multimodal distributions \citep{fernandez2025}.
 
In many applications, however, orientation observations are recorded or reported on a finite set of angular bins rather than as exact continuous angles. Such binned orientation data may arise from rounding before reporting, limited angular resolution of the measuring instrument, or a data-collection design in which observations are recorded directly in prespecified categories.  

One example arises in geology, where the orientations of the shortest projection elongations of sand grains from the Recent Gulf Coast beach are observed as axial data and grouped into $18$ angular bins, each of width $10^\circ$ \citep{curray1956}. These observations therefore provide a natural example of discrete axial data recorded on a finite angular support.

An example with $4$--fold rotational structure arises in geography from the orientation of urban street networks. Street-segment bearings are usually treated as axial because a street represented by an undirected edge has no distinguished starting direction. In addition, planned cities such as Manhattan are organized approximately as orthogonal grids. Such networks contain four equivalent full-circle directions separated by $90^\circ$. After reduction modulo $90^\circ$, their orientations can be represented on $[0,90^\circ)$ and grouped into a finite number of angular bins \citep{boeing2019}.

Binned orientation data also occur in other settings. For example, the swimming directions of \textit{Daphnia} relative to the plane of polarisation of surrounding light have been reported in $18$ axial bins of width $10^\circ$ \citep{waterman1963}. In human-detection applications using the Histogram of Oriented Gradients (HOG) descriptor, image gradients are commonly treated as unsigned, so their orientations are represented over $[0^\circ,180^\circ)$ and accumulated into a finite number of angular bins, commonly nine, to construct local orientation histograms \citep{dalal2005}.

The existing continuous distributions can be applied to such binned observations by integrating their densities over the corresponding angular bins. This process yields a valid probability mass for these angular bins, but numerical integrations are required because most of the density functions do not provide a closed form for these integrations, which are often repeatedly required during parameter estimation. An alternative is to formulate a probability mass function directly on observed angular bins. Such native discrete models assign probabilities to these observed angular bins directly and do not require any numerical integration.

In this paper, we develop a unified framework for constructing native discrete distributions on a reduced angular support, where directions differing by multiples of $2\pi/p$ are treated as equivalent. Two general construction methods are introduced, which are based on the symmetrization and folding of a discrete circular distribution, and wrapping of a discrete linear distribution onto a discrete angular grid in $[0, 2\pi/p)$.

For the axial sand-grain orientation data, the proposed models are compared with discretized versions of suitable continuous axial distributions, whose bin probabilities are obtained through numerical integration. The proposed models attain smaller Akaike information criterion values than the discretized continuous models considered. The framework is then applied to binned Manhattan street-orientation data to demonstrate its use for higher-order rotational structure, specifically the $4$--fold case.

The paper is organized as follows. Section~\ref{sec:2} presents the general framework for constructing new discrete models. Section~\ref{sec:3} introduces particular distributions arising from the two constructions and discusses their properties and estimation methods. Section~\ref{sec:5} presents the simulation study. Section~\ref{sec:6} applies the proposed models to the sand-grain and Manhattan street-orientation datasets. Finally, Section~\ref{sec:7} concludes this paper and outlines possible directions for future research.

\section{A General Framework for Discrete Distributions with $p$--Fold Rotational Symmetry} \label{sec:2}
In this section, we present two general methodologies for constructing
discrete probability models for observations exhibiting $p$-fold rotational
symmetry. Throughout this article, directions on the full circle that differ by
integer multiples of $2\pi/p$ are regarded as equivalent.

\subsection{Symmetrization and Folding of a Discrete Circular Variable}\label{sub.2.1}

Suppose $\Phi$ is a discrete circular random variable with a pmf $P_\Phi(\phi; \alpha,\omega)$ defined on the support 

\[
\Omega_\Phi = \left\{\frac{2\pi k}{\mathrm{m}}: k=0,\ldots,\mathrm{m}-1\right\}\subset[0,2\pi),
\]
consisting of $\mathrm{m}$ equally spaced angles on the unit circle, where $\mathrm{m}$ is a positive integer. Here $\alpha \in [0, 2\pi)$ denotes a rotational location parameter, and $\omega$ denotes the collection of any remaining model parameters.   

We construct a pmf for observations exhibiting $p$--fold rotational symmetry
by symmetrizing and folding the distribution of $\Phi$. For a positive integer
$p$, define the reduced angular support by mapping the states of
$\Omega_\Phi$ modulo $2\pi/p$:
\[
\Omega_\Theta
=
\left\{
\phi \bmod \frac{2\pi}{p}:
\phi\in\Omega_\Phi
\right\}.
\]
In general, this finite support contains
$\mathrm{m}/\operatorname{gcd}(\mathrm{m},p)$ equally spaced angular points,
where $\operatorname{gcd}(\mathrm{m},p)$ denotes the greatest common divisor
of $\mathrm{m}$ and $p$. The proposed construction is defined on this reduced
support $\Omega_\Theta$ and does not require $p$ to divide $\mathrm{m}$.
For the practical applications considered in this article, however, we take
$\mathrm{m}=\mathrm{h}p$,
where $\mathrm{h}$ is a positive integer. Since $p$ divides $\mathrm{m}$, the reduced
support then contains exactly $\mathrm{h}=\mathrm{m}/p$ angular points and can be
written explicitly as
\[
\Omega_\Theta
=
\left\{
\theta_j=\frac{2\pi j}{\mathrm{m}}:
j=0,\ldots,\mathrm{h}-1
\right\}
\subset
\left[0,\frac{2\pi}{p}\right).
\]
For each $\theta_j\in\Omega_\Theta$, the corresponding directions in the
original support are
\[
\theta_j+r\frac{2\pi}{p}
=
\frac{2\pi(j+r\mathrm{h})}{\mathrm{m}},
\qquad
r=0,\ldots,p-1,
\]
all of which belong to $\Omega_\Phi$. Thus, each point of $\Omega_\Theta$
represents $p$ directions in $\Omega_\Phi$ that differ by integer multiples
of $2\pi/p$. Accordingly, the proposed pmf is defined on the reduced discrete
support $\Omega_\Theta$, with these $p$ rotationally equivalent directions
represented by a single point.

To obtain the probability mass on $\Omega_\Theta$, we first define a function $ P^*_{p,\Phi}(x)$ on $[0, 2\pi)$ as 
\begin{equation}\label{eq:2.2}
 P^*_{p,\Phi}(x) = 
\begin{cases}
\displaystyle \frac{1}{p} \sum_{s = 0}^{p - 1} P_\Phi\left(x; \left(\alpha + s \frac{2 \pi}{p} \right) \bmod 2\pi, \omega \right), & x \in \Omega_\Phi \\[8pt]
0, & x\notin \Omega_\Phi.  
\end{cases}
\end{equation} 
For each $s=0,\ldots,p-1$, the corresponding term in \eqref{eq:2.2}
represents the underlying circular model evaluated at the shifted location
parameter
\[
\left(\alpha+s\frac{2\pi}{p}\right)\bmod 2\pi.
\]
The modulo $2\pi$ operation ensures that the shifted location remains in
$[0,2\pi)$. Thus, for $x\in\Omega_\Phi$, $P^*_{p,\Phi}(x)$ averages the
$p$ pmfs corresponding to these shifted values of the location parameter,
while the discrete support $\Omega_\Phi$ remains unchanged. For
$x\notin\Omega_\Phi$, $P^*_{p,\Phi}(x)$ is set equal to zero. The pmf on $\Omega_\Theta$ is defined as
\begin{equation} \label{eq:2.3}
    P_\Theta(\theta) = \sum_{r=0}^{p - 1} P^*_{p,\Phi} \left(\theta + r\frac{2 \pi}{p}\right), \quad \text{for } \theta \in \Omega_\Theta.
\end{equation} 

\begin{proposition}
The function $P_\Theta$ defined in \eqref{eq:2.3} is a valid probability mass
function on $\Omega_\Theta$.
\end{proposition}

\begin{proof}
From the definition in \eqref{eq:2.2}, each term in the sum given in \eqref{eq:2.3} is non-negative. Hence, $P_\Theta(\theta) \geq 0$ for every $\theta\in\Omega_\Theta$. 
It remains to verify that the probability masses sum to one. We have
\begin{equation}\label{eq:2.4}
\sum_{\theta \in \Omega_\Theta} P_\Theta(\theta) = \sum_{\theta \in \Omega_\Theta} \sum_{r=0}^{p - 1} P^*_{p,\Phi} \left(\theta + r\frac{2 \pi}{p}\right).  
\end{equation}
For every $\phi \in \Omega_\Phi$, let
\[ \theta = \phi \bmod \frac{2\pi}{p}.
\]
Then $\theta \in \Omega_\Theta$, and there exists a unique $r \in \left\{0,\cdots,p-1\right\}$ such that
\[
\phi = \theta + r\frac{2\pi}{p}.
\]
Any value $\theta + r\frac{2\pi}{p} \notin \Omega_\Phi$ contributes zero because $P^*_{p,\Phi}(x) = 0$ for $x \notin \Omega_{\Phi}$. Therefore the double sum can be written as

\begin{align*}
 \sum_{\theta \in \Omega_\Theta} P_\Theta(\theta)& =  \sum_{\phi \in \Omega_\Phi} P^*_{p,\Phi}(\phi)\\
 & =  \sum_{\phi \in \Omega_\Phi} \frac{1}{p}\left( \sum_{s = 0}^{p - 1} P_\Phi\left(\phi; \left(\alpha + s \frac{2 \pi}{p} \right) \bmod 2\pi, \omega \right)\right) \\
 & = \frac{1}{p} \sum_{s = 0}^{p - 1} \sum_{\phi \in \Omega_\Phi}  P_\Phi\left(\phi; \left(\alpha + s \frac{2 \pi}{p} \right) \bmod 2\pi, \omega \right).
\end{align*}
Since both sums are finite, we interchanged their order. For each $s \in \left\{0,\cdots,p-1\right\}$, the corresponding $P_\Phi$ is a valid pmf over $\Omega_\Phi$. Hence,
\[
\sum_{\theta \in \Omega_\Theta} P_\Theta(\theta) =
\frac{1}{p}\sum_{s=0}^{p-1}1 = 1.
\]

Therefore, $P_\Theta$ is a valid pmf on $\Omega_\Theta$.
\end{proof} 
The above framework is based on symmetrization by shifting the rotational
location parameter $\alpha$ by multiples of $2\pi/p$ and averaging the
resulting $p$ pmfs, followed by folding the symmetrized distribution onto
the reduced discrete angular support $\Omega_\Theta$. Under this construction, directions in $\Omega_\Phi$ that differ by an integer multiple of $2\pi/p$ are mapped to the same point in $\Omega_\Theta$.


\noindent\textbf{Note.} The particular representation of 
$\Omega_\Phi \subset [0,2\pi)$ is adopted for convenience. More generally, 
the underlying discrete circular distribution may be defined on a rotated 
equally spaced circular support
\[
\left\{
\delta+\frac{2\pi k}{m}:k=0,\ldots,m-1
\right\},
\]
where the angles are interpreted modulo $2\pi$ and may be represented on any interval of length $2\pi$, such as $[0,2\pi)$ or $(-\pi,\pi]$. Here, $\delta$ determines the location of the first grid point, with successive support points separated by $2\pi/\mathrm{m}$.


\subsection{Wrapping of a Linear Variable}\label{sub.2.2}
Let $\mathbb{Z}$ denotes the set of integers, and let $X$ be a discrete random variable defined on $\mathbb{Z}$ with the pmf $P_X$. Then, the corresponding discrete wrapped pmf for the angular random variable $\Theta$, defined on the support $\Omega_\Theta = \left\{ \frac{2\pi k}{p \mathrm{m}}: k \in \mathbb{Z}_\mathrm{m} \right\}$, where $ \mathbb{Z}_{\mathrm{m}}$ is a cyclic group of integers of order $\mathrm{m}$, is constructed as follows: 

\begin{equation}\label{eq:5.1}
 P_{\Theta}\left(\frac{2 \pi k}{p \mathrm{m}}\right) = \sum_{j = -\infty}^{\infty} P_{X}(k + \mathrm{m}j).   
\end{equation}
 The function given in \eqref{eq:5.1} is non-negative, and its sum over the entire support evaluates to
$$ \sum_{k = 0}^{\mathrm{m}-1} P_{\Theta}\left(\frac{2 \pi k}{p \mathrm{m}}\right) = \sum_{k = 0}^{\mathrm{m}-1} \sum_{j = -\infty}^{\infty} P_{X}(k + \mathrm{m}j) = \sum_{x = -\infty}^{\infty} P_{X}(x) = 1. $$
Therefore, $P_\Theta$ is a valid pmf for $\Theta$.

Here, it is noteworthy that the probability masses depend only on the grid size $\mathrm{m}$ and the underlying pmf $P_X$. Specifically, $\Omega_{\Theta}$ consists of $\mathrm{m}$ equally spaced categories that lie in $[0,2\pi/p)$, and, for fixed $\mathrm{m}$ and $P_X$, changing $p$ changes the locations and spacing of these angular categories, while the probability masses associated with the category indices $k\in\mathbb{Z}_{\mathrm{m}}$ remain unchanged.

The two constructions presented in this section provide general methods for obtaining discrete distributions on finite angular supports for modelling observations exhibiting $p$-fold rotational symmetry. In the following section, particular distributions obtained from these constructions are considered.

\section{Particular Discrete $p$--Fold Models}\label{sec:3}
This section presents distributions derived from the methodology developed in Section~\ref{sec:2}. Their special cases and corresponding parameter estimation methods are also discussed.

\subsection{Generalized Wrapped Geometric-Based $p$-Fold Distribution}\label{sub:3.1}
 The generalized wrapped geometric distribution (GWG), introduced by \citet{Jha2018}, is given by

\begin{equation}\label{eq:3.1}
P_{\Phi}\left(\frac{2 \pi k}{\mathrm{m}}; a, q\right)  = C \big(q^\zeta + q^{\mathrm{m}-\zeta}\big),
\end{equation}
where \[C = \frac{(1 - q)}{\left(q^b + q^{1-b}\right)\left(1 - q^{\mathrm{m}}\right)},\]
$\zeta = (k-a)\bmod\mathrm{m},~ a \in [0, \mathrm{m})$, and $b$ is the fractional part of $a$, i.e., $a - \lfloor a \rfloor$, and $\mathrm{m}$ is a positive integer. Here the pmf $P_{\Phi}$ is defined on the set $\{2 \pi k/\mathrm{m} : k \in \mathbb{Z}_{\mathrm{m}}\}$, and its mode occurs at either $2\pi \lfloor a \rfloor/\mathrm{m}$ or $2\pi \lfloor a+1 \rfloor/\mathrm{m}$ or both.

Consider the reduced angular support

\[\Omega_\Theta = \left\{ \left( \frac{2 \pi k}{\mathrm{m}} \right) \bmod \frac{2\pi}{p} : k \in \mathbb{Z}_{\mathrm{m}} \right\}, \]
where $\mathrm{m} = \mathrm{h}p$. $\Omega_\Theta$ consists of $\mathrm{h}$ equally spaced angular points. 
Since $\mathrm{m}=\mathrm{h}p$, this support can be represented as
\[\Omega_\Theta = \left\{
\theta_j=\frac{2\pi j}{\mathrm{m}}:
j=0,\ldots,\mathrm{h}-1 \right\}, \]
and therefore consists of $\mathrm{h}$ equally spaced angular points.
Moreover, for every $\theta_j\in\Omega_\Theta$ and
$r=0,\ldots,p-1$,
\[
\theta_j+r\frac{2\pi}{p} = \frac{2\pi(j+r\mathrm{h})}{\mathrm{m}}
\in\Omega_\Phi.
\]
Therefore, applying the construction in Section~\ref{sub.2.1}, the pmf on
$\Omega_\Theta$ is given by

\begin{equation*}
P_\Theta(\theta_j) = \frac{1}{p} \sum_{r=0}^{p-1} \sum_{s=0}^{p-1}
P_\Phi\left(\theta_j+r\frac{2\pi}{p};
\left(a+s\frac{\mathrm{m}}{p}\right)\bmod\mathrm{m}, q\right).
\end{equation*} 
Let $\eta_j = (j-a) \bmod \mathrm{h}$ and $a =\lfloor a\rfloor + b$, then the closed-form pmf on $\Omega_\Theta$ is given by

\begin{equation} \label{eq:3.2} 
P_\Theta(\theta_j) = C_{\mathrm{h}}\left(q^{\eta_j}+q^{\mathrm{h}-\eta_j}\right),
\end{equation}
where $\mathrm{h}=\mathrm{m}/p$, and
\[
C_{\mathrm{h}} = \frac{1-q}
{\left(q^b+q^{1-b}\right)\left(1-q^{\mathrm{h}}\right)}.
\]    
We refer to the distribution in \eqref{eq:3.2} as the GWG-based $p$-fold (GWGB-$p$) distribution obtained by symmetrization and folding of a GWG distribution. Figure~\ref{fig:1.1} provides a schematic illustration of these two steps. This distribution is appropriate for modelling data that exhibit $p$--fold rotational symmetry.

\begin{figure}
    \centering
    \includegraphics[width= 1\linewidth]{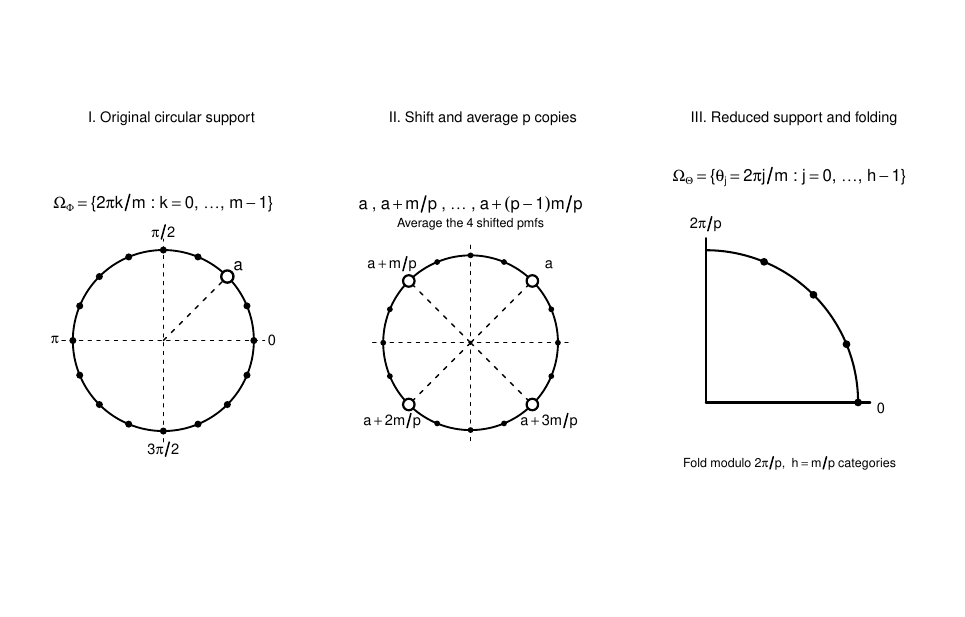}
    \vspace{-2cm}
   \caption{Schematic representation of the symmetrization and folding construction for a discrete circular distribution with $p=4,~\mathrm{m} = 16$ and $a = 2$. The rotational location parameter is shifted by multiples of $2\pi/p$, the resulting four shifted pmfs are averaged, and the symmetrized distribution is folded onto the reduced discrete angular support $\Omega_{\Theta}$, which contains $\mathrm{h}=\mathrm{m}/p$ grid points.}
    \label{fig:1.1}
\end{figure}

\subsubsection{Discussion of the Model Characteristics}\label{sec:3.1.1}
The model is periodic in $a$ with period $\mathrm{m}/p$, therefore, without loss of generality, we restrict $a\in[0,\mathrm{h})$. The principal modal characteristic of the underlying distribution is preserved, that is, the pmf in \eqref{eq:3.2} attains its largest probability mass at one of the two neighboring directions of the principal direction $2\pi a/\mathrm{m}$. Furthermore, if $a$ is an integer, then it attains its highest probability mass exactly at this principal direction. This preservation follows directly from the proposed framework. This is because the proposed framework does not merely create multiple copies of the underlying distribution by rotating only the variable $\phi$, but also generates corresponding copies with location parameters shifted by multiples of $\mathrm{m}/p$.  


\noindent\textbf{Limiting cases.}
The limiting behavior of the distribution provides a useful interpretation of the concentration parameter $q$.

\begin{itemize}

    \item As $q \to 1$,
    \[
    P_\Theta(\theta_j)\longrightarrow \frac{1}{h},
    \qquad j=0,\ldots,h-1.
    \]
    Thus, the distribution converges to the discrete uniform distribution
    on the reduced angular support.

    \item As $q \to 0$, if $a$ is an integer, then
    \[
    \lim_{q\to0}P_\Theta(\theta_j)
    =
    \begin{cases}
        1, & j=a \pmod h,\\
        0, & \text{otherwise}.
    \end{cases}
    \]
    Hence, the distribution becomes degenerate at the principal direction
    $2\pi a/m$.

    \item As $q \to 0$, if $a$ is non-integer, write
\[
a=\ell+b,
\qquad
\ell=\lfloor a\rfloor,
\qquad
0<b<1.
\]
In this case, only the two categories neighboring $a$ can retain
positive limiting probability. Let
\[
\lambda
=
\lim_{q\to0}
\frac{q^{\,1-b}}
     {q^b+q^{\,1-b}}.
\]
Then
\[
\lim_{q\to0}P_\Theta(\theta_j)
=
\begin{cases}
1-\lambda, & j=\ell,\\[1mm]
\lambda, & j=(\ell+1)\bmod h,\\[1mm]
0, & \text{otherwise}.
\end{cases}
\]
Since
\[
\lambda
=
\begin{cases}
0, & 0<b<\dfrac{1}{2},\\[2mm]
\dfrac{1}{2}, & b=\dfrac{1}{2},\\[2mm]
1, & \dfrac{1}{2}<b<1,
\end{cases}
\]
the limiting distribution can equivalently be written as
\[
\lim_{q\to0}P_\Theta(\theta_j)
=
\begin{cases}
1, & j=\ell,\quad 0<b<\dfrac{1}{2},\\[2mm]
\dfrac{1}{2},
& j\in\{\ell,(\ell+1)\bmod h\},
  \quad b=\dfrac{1}{2},\\[2mm]
1, & j=(\ell+1)\bmod h,\quad
\dfrac{1}{2}<b<1,\\[2mm]
0, & \text{otherwise}.
\end{cases}
\]

Thus, the limiting distribution is concentrated at the discrete
category nearest to $a$. If $a$ lies exactly halfway between two
adjacent categories, the limiting probability is equally divided
between them.

\end{itemize}

The GWGB-$p$ distribution includes several important special cases. When $p=1$, GWGB-1 reduces to the underlying GWG distribution. When $p=2$, GWGB-2 is a discrete axial distribution exhibiting antipodal symmetry. More generally, GWGB-$p$ provides a discrete model for directional data exhibiting $p$-fold rotational symmetry.

\subsubsection{Parameter Estimation}
In this section, we consider methods for estimating the parameters of the
GWGB-$p$ distribution. Suppose that
$\theta_1,\theta_2,\ldots,\theta_n$ are a random sample of $n$ observations
from the GWGB-$p$ distribution with parameters $a$ and $q$. We consider the
following estimation methods.

\noindent
\textbf{Method of moments:} 
To estimate the parameters, we derive the population $p$th trigonometric moment
$\operatorname{E}(\operatorname{e}^{ip\Theta})$. Since
$\operatorname{e}^{ip\theta}=\operatorname{e}^{ip\left(\theta+2\pi/p\right)},$
this moment is appropriate for observations exhibiting $p$-fold rotational symmetry. The $p$th trigonometric moment of the GWGB-$p$ distribution given in \eqref{eq:3.2} is given by

\begin{equation}\label{eq:4.1}
\operatorname{E}(\operatorname{e}^{ip\Theta})
=\operatorname{e}^{i\omega (a-b)}\Bigg(\frac{(1-q)
\left[q^b-q^{2-b}+\operatorname{e}^{i\omega}\left(q^{1-b}-q^{1+b}\right)
\right]}{\left(1+q^2-2q\cos\omega\right)
\left(q^b+q^{1-b}\right)}
\Bigg),
\end{equation}
where
$\omega= 2\pi p / \mathrm{m}$, and $b=a-\lfloor a\rfloor$.

Given a sample of $n$ observations exhibiting $p$-fold rotational symmetry, the sample $p$th trigonometric moment is given by

\begin{equation}\label{eq:4.2}
\bar{\phi}_p = \frac{1}{n}\sum_{j=1}^{n}\operatorname{e}^{ip\theta_j}.
\end{equation}

By the law of large numbers, the sample moment \eqref{eq:4.2} converges to the population moment \eqref{eq:4.1} as $n\rightarrow\infty$. The parameters $a$ and $q$ can be estimated simultaneously by equating the sample moment to the population moment,
\[\bar{\phi}_p = \operatorname{E}(\operatorname{e}^{ip\Theta}).\]
Equivalently, when the relevant moments are nonzero, the complex moment 
equation can be expressed in terms of its magnitude and argument. From 
\eqref{eq:4.1}, the magnitude of the population moment is 
\[
\left|\operatorname{E}(\operatorname{e}^{ip\Theta})
\right| =\frac{(1-q)\sqrt{\left(q^b-q^{2-b}\right)^2
+\left(q^{1-b}-q^{1+b}\right)^2+2\left(q^b-q^{2-b}\right)\left(q^{1-b}-q^{1+b}\right)\cos\omega}}{\left(1+q^2-2q\cos\omega\right)\left(q^b+q^{1-b}\right)
}.
\]

\noindent
Provided that 
$\operatorname{E}(\operatorname{e}^{ip\Theta})\neq0$, its argument is
\[
\operatorname{arg}
\left\{\operatorname{E}(\operatorname{e}^{ip\Theta})\right\}
=
\left(\omega a+\psi(q,b)\right)\bmod 2\pi,
\]
where 
\[
\psi(q,b) = -\omega b +\operatorname{atan2}
\left(\left(q^{1-b}-q^{1+b}\right)\sin\omega,\,
q^b-q^{2-b} + \left(q^{1-b}-q^{1+b}\right)\cos\omega \right).
\]
Here, $\operatorname{atan2}(y,x)$ denotes the four-quadrant arctangent, 
giving the argument of $x+iy$ in $(-\pi,\pi]$. When both the sample and 
population moments are nonzero, the method-of-moments estimates of $a$ and 
$q$ may be obtained by solving
$|\bar{\phi}_p| = \left|\operatorname{E}(\operatorname{e}^{ip\Theta})\right|$
and
$\operatorname{arg}(\bar{\phi}_p) = \left(\omega a+\psi(q,b)\right)\bmod 2\pi$
simultaneously, where $b=a-\lfloor a\rfloor$.

By the continuous mapping theorem \citep[Theorem~2.3]{van2000}, if
$\operatorname{E}(\operatorname{e}^{ip\Theta})\neq0$, then, as 
$n\rightarrow\infty$, the magnitude and circular phase of the sample moment 
converge to the corresponding magnitude and circular phase of the population 
moment. The method-of-moments estimator 
may be defined through the minimum-distance criterion
\[
(\hat a_{\mathrm{MoM}},\hat q_{\mathrm{MoM}})
=
\underset{a,q}{\operatorname{arg\,min}}\,
\left|
\bar{\phi}_p-
\operatorname{E}_{a,q}
\left(\operatorname{e}^{ip\Theta}\right)
\right|^2.
\]
Since the complex moment equation provides two real moment equations for the
two unknown parameters $a$ and $q$, whenever an exact solution exists in the
parameter space, the minimum value of the above criterion is zero and the
estimator is equivalent to solving the moment equations directly. Otherwise,
the criterion provides the closest admissible match to the sample moment.

 
\noindent
\textbf{Maximum likelihood estimation:}
 
Let
$
\frac{2\pi j_1}{\mathrm{m}},
\frac{2\pi j_2}{\mathrm{m}},
\ldots,
\frac{2\pi j_n}{\mathrm{m}},
\qquad
j_i\in\{0,\ldots,h-1\},\quad i=1,\ldots,n,
$
be a random sample from the GWGB-$p$ distribution, where
$h=\mathrm{m}/p$. The likelihood function is
\begin{align*}
L(a,q) = C_{\mathrm{h}}^{\,n} \prod_{i=1}^{n} \left( q^{\eta_i}+q^{\mathrm{h}-\eta_i} \right), 
\end{align*} 
where $ \eta_i=(j_i-a)\bmod \mathrm{h}.$ The log-likelihood function is

\begin{align*} 
\ell(a,q) ={}& n\log(1-q) - n\log\left(q^b+q^{1-b}\right) - n\log\left(1-q^{\mathrm{h}}\right) \nonumber\\ &\quad+ \sum_{i=1}^{n} \log\left( q^{\eta_i}+q^{\mathrm{h}-\eta_i} \right). 
\end{align*}
Since closed-form estimators are unavailable, the log-likelihood is maximized numerically over the parameter space, equivalently by minimizing the negative log-likelihood, using the \texttt{L-BFGS-B} algorithm implemented through the built-in \texttt{optim} function in \textsf{R}. This algorithm is used because it accommodates box constraints and is therefore well suited to the bounded parameter spaces of the proposed models.

Provided that the maximum likelihood estimate satisfies
\[
\hat a\in
\bigcup_{r=0}^{h-1}(r,r+1)
\qquad\text{and}\qquad
\hat q\in(0,1),
\]
so that the log-likelihood is locally differentiable at the estimate, the
covariance matrix of the maximum likelihood estimator is estimated by
\[
\widehat{\operatorname{Var}}(\hat{\boldsymbol{\tau}})
=
H_{\mathrm{NLL}}(\hat{\boldsymbol{\tau}})^{-1},
\]
where $\boldsymbol{\tau}=(a,q)^{\mathsf T}$ and
$H_{\mathrm{NLL}}(\hat{\boldsymbol{\tau}})$ denotes the Hessian matrix of
the negative log-likelihood evaluated at the maximum likelihood estimate.
The standard error of the $i$th parameter estimate is then
\[
\operatorname{SE}(\hat{\tau}_i)
=
\sqrt{
\left[
\widehat{\operatorname{Var}}(\hat{\boldsymbol{\tau}})
\right]_{ii}
}.
\]

\subsection{Skew Discrete Laplace-Based $p$-Fold Distribution}\label{sub:3.2}

The GWGB-$p$ distribution introduced in the preceding subsection uses a
common decay parameter on the two sides of its modal direction. To allow
the rates of decay to differ, we use the skew discrete Laplace distribution
on $\mathbb{Z}$ as the underlying linear distribution \citep{kozubowski2006}. Let $X$ be an integer-valued random variable with pmf,

\begin{equation*}
P_X(x;q_1,q_2) = \frac{(1-q_1)(1-q_2)}{1-q_1q_2}
\begin{cases}
 q_1^x, & x=0,1,2,\ldots,\\[1mm]
q_2^{-x}, & x=-1,-2,\ldots,
\end{cases}
\end{equation*}
where $q_1,~q_2\in(0,1)$. The parameters $q_1$ and $q_2$ control the geometric decay of the positive and negative tails, respectively. Applying the wrapping construction described in Section~\ref{sub.2.2} to the skew discrete Laplace distribution gives

\begin{align*}
P_\Theta\left(\frac{2\pi k}{p\mathrm{m}}\right)
&=
\sum_{j=-\infty}^{\infty}
P_X(k+mj;q_1,q_2)\\
&=
\frac{(1-q_1)(1-q_2)}{1-q_1q_2}\left\{
\sum_{j=0}^{\infty}q_1^{k+mj}
+
\sum_{j=0}^{\infty}q_2^{m-k+mj}
\right\}.
 \end{align*}
The closed--form expression of this distribution is as follows:
 \begin{equation}\label{eq:4.6}
 P_\Theta \left(\frac{2 \pi k}{p\mathrm{m}}\right) = \frac{(1-q_1)(1-q_2)}{1-q_1q_2} \left( \frac{q_1^k}{1-q_1^{\mathrm{m}}} + \frac{q_2^{\mathrm{m}-k}}{ 1 - q_2^{\mathrm{m}}} \right).
 \end{equation}
\autoref{fig:1.2} presents a schematic illustration of the wrapping construction corresponding to the distribution in \eqref{eq:4.6}.
Notably, when $p = 1$, the distribution given in \eqref{eq:4.6} is the wrapped skew Laplace distribution \citep{jayakumar2012}.  

\begin{figure}[!ht]
    \centering
    \includegraphics[width=1\linewidth]{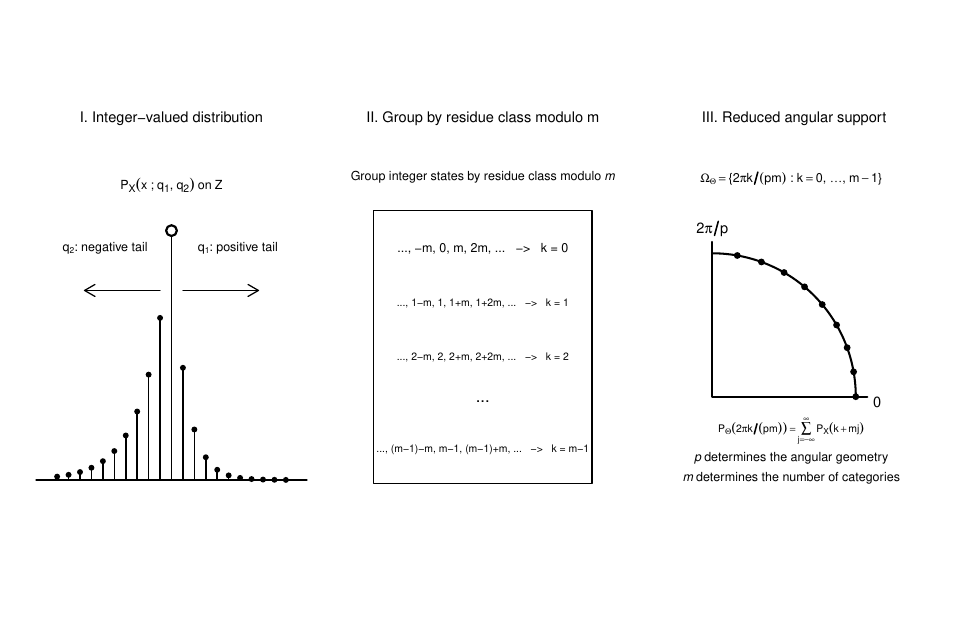}\vspace{-1.5cm}
    \caption{Schematic representation of the wrapping construction for a discrete linear distribution. Integer-valued states are grouped according to their residue classes modulo $\mathrm{m}$, and the probability masses within each residue class are summed and mapped to the corresponding category on the discrete angular support $\Omega_\Theta \subset [0,2\pi/p)$. Here, $\mathrm{m}$ determines the number of angular categories, whereas $p$ determines the angular geometry of the support. The illustration uses $p=4$, $\mathrm{m}=9$, $a=0$, $q_1=0.45$, and $q_2=0.65$.}
    \label{fig:1.2}
\end{figure}

The distribution in \eqref{eq:4.6} can be generalized by introducing a continuous shift parameter $a \in[0, \mathrm{m})$. It is formulated as follows:

\begin{equation}\label{eq:4.7}
  P_\Theta \left(\frac{2 \pi k}{p\mathrm{m}}\right) = C \left( \frac{q_1^{\zeta}}{1 - q_1^{\mathrm{m}}} + \frac{q_2^{\mathrm{m} - \zeta}}{1 - q_2^{\mathrm{m}}} \right), 
\end{equation} 
where $\zeta = (k - a) \bmod \mathrm{m}$, and $C$ is the normalizing constant. The reciprocal of $C$ is given by
\begin{equation*}
  \frac{1}{C} =  \frac{q_1^{1-b}}{1-q_1} +  \frac{q_2^{b}}{1-q_2},
\end{equation*}
where $b$ is the fractional part of $a$.
We refer to the distribution given in \eqref{eq:4.7} as the skew discrete Laplace-based $p$-fold (SDLB-$p$) distribution. For example, SDLB-$2$
and SDLB-$4$ denote the axial and four-fold cases, respectively.

When $q_1=q_2=q$, the pmf in \eqref{eq:4.7} reduces to
\[
P_\Theta\left(\frac{2\pi k}{p\mathrm{m}}\right)
=
\frac{1-q}
     {(q^b+q^{1-b})(1-q^{\mathrm{m}})}
\left(
q^\zeta+q^{\mathrm{m}-\zeta}
\right).
\]
It is noteworthy that, if the number $\mathrm{m}$ of reduced categories in the
SDLB-$p$ construction is identified with $h=\mathrm{m}/p$ in the GWGB-$p$
construction, then the two distributions are defined on the same reduced angular
support. Under this identification, they have the same probability mass function
when $q_1=q_2=q$. 

\subsubsection{Discussion of the Model Characteristics} \label{sub:3.2.1}
To study the modal structure, we define 
\[
g(x) = \frac{q_1^x}{1-q_1^\mathrm{m}} + \frac{q_2^{\mathrm{m}-x}}{1-q_2^\mathrm{m}},
\qquad x\in [0, \mathrm{m}].
\]
Since $q_1,~q_2 \in (0, 1)$, the second derivative of $g$ on $(0, \mathrm{m})$ is 
\[
g''(x) = \frac{(\log q_1)^2q_1^x}{1-q_1^m}
  + \frac{(\log q_2)^2q_2^{m-x}}{1-q_2^m}>0.
\]

Therefore, $g$ is strictly convex on $[0, \mathrm{m}]$, and hence its maximum is attained at one (or both) of the endpoints. Since $g(0) = g(\mathrm{m})$, it follows that $g(x) < g(0)$ for all values $x \in (0, \mathrm{m})$. Consequently, the pmf attains its largest probability mass at the direction corresponding to $\zeta = 0$, namely the principal direction $2\pi a/p\mathrm{m}$, whenever $a$ is an integer. For non-integer values of $a$, $\zeta = 0$ is unattainable, and thus the largest probability mass is attained at one or both of the two neighboring directions of $2\pi a/p\mathrm{m}$. 

The strict convexity of $g$ also determines the antimodal structure. In the symmetric case $q_1=q_2=q$,
\[
g(x)=\frac{q^x+q^{\mathrm{m}-x}}{1-q^{\mathrm{m}}},
\]
and its first derivative is
\[
g'(x)=\frac{\log q}{1-q^{\mathrm{m}}}
\left(
q^x-q^{\mathrm{m}-x}
\right).
\]
The critical point is obtained by setting $g'(x)=0$, which gives $x = \mathrm{m}/2$.
Because $g$ is strictly convex, this critical point is the unique global
minimum of $g$ on $[0,\mathrm{m}]$.
If $\mathrm{m}/2$ is attainable among the values
\[
\zeta=(k-a)\bmod \mathrm{m},
\qquad k\in\mathbb{Z}_{\mathrm{m}},
\]
then the pmf attains its minimum at $\zeta=\mathrm{m}/2$. A displacement of
$\mathrm{m}/2$ in the category index corresponds to an angular displacement
\[
\frac{2\pi}{p\mathrm{m}}
\frac{\mathrm{m}}{2}
=
\frac{\pi}{p}.
\]
Hence, in this case, the antimodal direction is
\[
\mu_{\mathrm{anti}}
=
\left(
\frac{2\pi a}{p\mathrm{m}}
+
\frac{\pi}{p}
\right)
\bmod
\frac{2\pi}{p}.
\]

If $\mathrm{m}/2$ is not attainable, the minimum probability mass is attained
at the category or categories whose corresponding values of $\zeta$ are
closest to $\mathrm{m}/2$.

The parameters $q_1$ and $q_2$ govern the decay of the probability mass
in the two rotational directions away from the principal direction. For the
first term,
\[
\frac{q_1^{x+1}/(1-q_1^\mathrm{m})}
     {q_1^x/(1-q_1^\mathrm{m})} = q_1,
\]
and for the second term,
\[
\frac{q_2^{\mathrm{m}-(x-1)}/(1-q_2^\mathrm{m})}
     {q_2^{\mathrm{m}-x}/(1-q_2^\mathrm{m})} = q_2.
\]
Thus, $q_1$ and $q_2$ primarily control the geometric decay of the probability mass on the two sides of the principal direction. Smaller values of $q_i$ produce faster decay on the corresponding side and greater concentration around the principal direction, whereas values closer to one produce slower decay and greater dispersion.

Depending on the value of $p$, the SDLB-$p$ distribution is suitable for modelling several types of directional data. When $p=1$, SDLB-$1$ is the discrete circular distribution. When $p=2$, SDLB-$2$ is the discrete axial distribution appropriate for modelling discrete axial data exhibiting antipodal symmetry. More generally,
SDLB-$p$ is appropriate for discrete directional data exhibiting $p$-fold
rotational symmetry.

\noindent\textbf{Limiting cases.}
Since $q_1,q_2\in(0,1)$, we evaluate the limits of the pmf as $(q_1,q_2)$ approaches the boundary points $(0,0)$, $(0,1)$, $(1,0)$, and $(1,1)$.

\begin{itemize}

\item If $q_1\to1$ or $q_2\to1$, then
\[P_\Theta\left(\frac{2\pi k}{p\mathrm{m}}\right)
\longrightarrow \frac{1}{\mathrm{m}},
\qquad k=0,\ldots,\mathrm{m}-1.\]
Thus, the distribution converges to the discrete uniform distribution on the reduced angular support.

\item As $(q_1,q_2)\to(0,0)$, if $a$ is an integer, then
\[\lim_{(q_1,q_2)\to(0,0)}
P_\Theta\left(\frac{2\pi k}{p\mathrm{m}}\right)
= \begin{cases}
    1, & k=a\bmod \mathrm{m},\\
    0, & \text{otherwise}.
\end{cases}\]
Hence, the distribution becomes degenerate at the principal direction $2\pi a/(p\mathrm{m})$.

\item As $(q_1,q_2)\to(0,0)$, if $a$ is non-integer, write
    \[ a=\ell+b,
    \qquad \ell=\lfloor a\rfloor,
    \qquad 0<b<1. \]
In this case, only the two categories neighboring $a$ can retain positive limiting probability. If
    \[\lambda =
    \lim_{(q_1,q_2)\to(0,0)}
    \frac{q_1^{\,1-b}}
         {q_1^{\,1-b}+q_2^{\,b}}\]
exists, then
    \[\lim_{(q_1,q_2)\to(0,0)}
    P_\Theta\left(\frac{2\pi k}{p\mathrm{m}}\right)
    =\begin{cases}
    1-\lambda, & k=\ell,\\[1mm]
    \lambda, & k=(\ell+1)\bmod \mathrm{m},\\[1mm]
    0, & \text{otherwise}.
    \end{cases}\]
Therefore, for non-integer $a$, the limiting distribution at $(q_1,q_2)=(0,0)$ generally depends on the relative rates at which $q_1$ and $q_2$ approach zero.

In particular, if $q_1=q_2=q\to0$, the limiting distribution reduces to
\[\lim_{q\to0}
P_\Theta\left(\frac{2\pi k}{p\mathrm{m}}\right)
 = \begin{cases}
 1,
 & k=\ell,\quad 0<b<\dfrac{1}{2},\\[2mm]
 \dfrac{1}{2},
 & k \in \{\ell, (\ell+1)\bmod \mathrm{m} \}\quad b=\dfrac{1}{2},\\[2mm]
 1,
 & k=(\ell+1)\bmod \mathrm{m},\quad \dfrac{1}{2}<b<1,\\[2mm]
 0,
 & \text{otherwise}.
 \end{cases} \]
Thus, along the equal-rate path, the distribution concentrates on
the category nearest to $a$, with equal limiting probability on the two neighboring categories when $b=1/2$.

\end{itemize}

\subsubsection{Parameter Estimation}  

\textbf{Method of Moments Estimation.}
The trigonometric moments can be calculated to facilitate parameter estimation via the method of moments. For observations exhibiting $p$--fold rotational symmetry, the appropriate complex representation is $Z=e^{ip\Theta}$, because angles differing by $2\pi/p$ have the same value of $Z$. For a positive integer $r$, the $r$th $p$--fold trigonometric moment is given by

\begin{equation}\label{eq:2.3.1}
\operatorname{E}(Z^r) = C e^{i\gamma_r a}\left(q_1^{1-b}
\frac{e^{i\gamma_r(1-b)}}{1-q_1e^{i\gamma_r}}+q_2^b
\frac{e^{-i\gamma_r b}}{1-q_2e^{-i\gamma_r}}\right),
\end{equation}
where
$\gamma_r=\frac{2\pi r}{\mathrm{m}},$ and $b$ is the fractional part of $a$.

Let
$2\pi k_1/p\mathrm{m},2\pi k_2/p\mathrm{m},\dots,2\pi k_n/p\mathrm{m}$
be a random sample of $n$ observations on the discrete grid. The sample $r$th $p$--fold trigonometric moment is given by
\begin{equation*}
\hat{\phi}_r =
\frac{1}{n}\sum_{j=1}^n e^{i\gamma_r k_j} = \bar{C}_r+i\bar{S}_r,
\end{equation*}
where
\[
\bar{C}_r = \frac{1}{n}\sum_{j=1}^n\cos(\gamma_r k_j)
\quad\text{and}\quad
\bar{S}_r = \frac{1}{n}\sum_{j=1}^n\sin(\gamma_r k_j).
\]

For all values of $q_1$, $q_2$, and $a$, the moments given in
\eqref{eq:2.3.1} are finite. By the law of large numbers, the sample
moments converge to their corresponding population moments. Since the model contains three unknown parameters, we restrict attention to $\mathrm{m}\geq4$, and a minimum-distance method of moments estimator based on the first two $p$--fold trigonometric moments may be defined as
\begin{equation}
\widehat{\boldsymbol{\tau}}_{\mathrm{MoM}}
=
\arg\min_{a,q_1,q_2}
\sum_{r=1}^{2}
\left[
\left\{\bar{C}_r-\alpha_r(\boldsymbol{\tau})\right\}^{2}
+
\left\{\bar{S}_r-\beta_r(\boldsymbol{\tau})\right\}^{2}
\right],
\end{equation}
where
\[
\boldsymbol{\tau}=(a,q_1,q_2)^{\mathsf T},
\qquad
\operatorname{E}(Z^r)
=
\alpha_r(\boldsymbol{\tau})
+
i\beta_r(\boldsymbol{\tau}).
\]
Owing to the nonlinear form of the objective function, the minimization would need to be carried out numerically.

\noindent
\textbf{Maximum Likelihood Estimation.}
As an alternative to the method of moments, we consider maximum likelihood estimation (MLE). Given a sample of $n$ independent observations, the likelihood function is formulated as:

\begin{equation}
  L(\boldsymbol{\tau}) = \prod_{j=1}^n C \left( \frac{q_1^{\zeta_j}}{1 - q_1^{\mathrm{m}}} + \frac{q_2^{\mathrm{m} - \zeta_j}}{1 - q_2^{\mathrm{m}}} \right),
\end{equation}
where $\zeta_j = (k_j-a)\bmod \mathrm{m}$ and $\boldsymbol{\tau}=(a,q_1,q_2)^{\mathsf T}$. The corresponding log-likelihood function, $\ell(\boldsymbol{\tau}) = \log L(\boldsymbol{\tau})$, is expressed as:
\begin{equation}
  \ell(\boldsymbol{\tau}) = n \log C + \sum_{j=1}^n \log \left( \frac{q_1^{\zeta_j}}{1 - q_1^{\mathrm{m}}} + \frac{q_2^{\mathrm{m} - \zeta_j}}{1 - q_2^{\mathrm{m}}} \right).
\end{equation}
Due to the complex structural form of the log-likelihood function, the score equations do not yield closed-form estimators. Therefore, we estimate the parameters by numerically maximizing the log-likelihood function.


\section{Simulation}\label{sec:5}
To evaluate the finite-sample performance of the proposed estimation methods, 
we conduct a simulation study for the GWGB-$p$ and SDLB-$p$ distributions using both the method of moments (MoM) and maximum likelihood estimation (MLE). Samples of sizes $n=200$ and $500$ are generated under different combinations of the model parameters and grid sizes. For each setting, the results are based on $N=1000$ independent Monte Carlo replications. The empirical mean and standard deviation of the parameter estimates are computed across the replications to assess the accuracy of the two estimation methods.

For the GWGB-$p$ model, the category probabilities depend on $p$ and $\mathrm{m}$
through $\mathrm{h}=\mathrm{m}/p$. Thus, settings with the same value of $\mathrm{h}$ have the same probability structure but different angular interpretations. Similarly, for
the SDLB-$p$ model, the category probabilities remain unchanged when $p$
varies while $\mathrm{m}$ and the model parameters are fixed. In this case, $p$
changes only the angular embedding of the categories.

For the linear parameters $q$, $q_1$, and $q_2$, the empirical mean and 
standard deviation are calculated using the usual arithmetic summaries. 
However, the location parameter $a$ is periodic rather than linear, and it lies 
in an interval $[0, L)$, where $L=\mathrm{h}=\mathrm{m}/p$ for the GWGB-$p$ distribution and $L=\mathrm{m}$ for the SDLB-$p$ distribution.

Since the location parameter $a$ is periodic on $[0,L)$, its empirical
mean and standard deviation were calculated using circular statistics.
For the $i$th simulation replicate, let $\hat a_i$ denote the estimate
of $a$, and define the corresponding angular value as
\[
\theta_i=\frac{2\pi\hat a_i}{L},
\qquad i=1,\ldots,N.
\]
The mean resultant components and resultant length are
\[
\bar C=\frac{1}{N}\sum_{i=1}^{N}\cos\theta_i,
\qquad
\bar S=\frac{1}{N}\sum_{i=1}^{N}\sin\theta_i,
\qquad
\bar R=\sqrt{\bar C^2+\bar S^2}.
\]
The circular mean and circular standard deviation, expressed on the
original scale of $a$, are then given by
\[
\bar a
=
\left[
\frac{L}{2\pi}
\operatorname{atan2}(\bar S,\bar C)
\right]\bmod L,
\qquad
\mathrm{SD}_a
=
\sqrt{-2\log\bar R}.
\]
Here, $L$ denotes the period of the location parameter $a$.

Under this periodicity, $0$ and $L$ represent the same location. Hence, when
the true value is $a=0$, an empirical circular mean close to $L$ should be
interpreted as lying close to $0$ across the periodic boundary, rather than as
indicating a large estimation error.

The simulation results are presented in Tables \ref{tab:4.1}--\ref{tab:4.4}. 
Tables \ref{tab:4.1} and \ref{tab:4.2} report the maximum likelihood 
estimation results for the GWGB-$p$ and SDLB-$p$ distributions, respectively. Tables \ref{tab:4.3} and \ref{tab:4.4} present the corresponding method of moments results for the GWGB-$p$ and SDLB-$p$ distributions, respectively.
For each estimation method, the tables present the empirical mean and standard 
deviation of the parameter estimates over the simulations, 
providing a direct assessment of the finite-sample accuracy of the MLE and MoM estimators. Tables \ref{tab:4.1}--\ref{tab:4.4} show that the variability of the parameter estimates decreases as the sample size increases. The precision of the estimate $\hat a$ depends strongly on the concentration parameters $q$, $q_1$, and $q_2$.

\begin{table}[!ht]
\caption{Maximum likelihood estimates for the GWGB-$p$ distribution.}
\label{tab:4.1}
\centering
\small
\renewcommand{\arraystretch}{0.85}

\begin{tabular}{ccccc}
\toprule
\multicolumn{2}{c}{True values} &
\multicolumn{1}{c}{Sample size} &
\multicolumn{2}{c}{MLE estimate and standard deviation} \\
\cmidrule(lr){1-2}
\cmidrule(lr){4-5}
$q$ & $a$ & $n$ & $\hat q$ & $\hat a$ \\
\midrule

\multicolumn{5}{c}{\textit{$p=2,\quad m=16$}} \\
\midrule

0.1 & 0.5 & 200 & 0.100 (0.019) & 0.500 (0.032) \\
    &     & 500 & 0.100 (0.013) & 0.501 (0.020) \\

0.3 & 0.5 & 200 & 0.299 (0.027) & 0.498 (0.060) \\
    &     & 500 & 0.300 (0.018) & 0.501 (0.037) \\

0.5 & 2.5 & 200 & 0.499 (0.039) & 2.503 (0.119) \\
    &     & 500 & 0.501 (0.025) & 2.502 (0.075) \\

0.7 & 2.5 & 200 & 0.699 (0.045) & 2.501 (0.387) \\
    &     & 500 & 0.706 (0.042) & 2.497 (0.258) \\

0.9 & 3.5 & 200 & 0.818 (0.048) & 3.519 (2.233) \\
    &     & 500 & 0.853 (0.042) & 3.624 (2.171) \\

\addlinespace[3pt]
\multicolumn{5}{c}{\textit{$p=4,\quad m=32$}} \\
\midrule

0.1 & 0.5 & 200 & 0.101 (0.021) & 0.499 (0.030) \\
    &     & 500 & 0.100 (0.013) & 0.502 (0.019) \\

0.3 & 0.5 & 200 & 0.300 (0.028) & 0.498 (0.060) \\
    &     & 500 & 0.299 (0.018) & 0.500 (0.039) \\

0.5 & 2.5 & 200 & 0.502 (0.033) & 2.497 (0.107) \\
    &     & 500 & 0.501 (0.026) & 2.498 (0.072) \\

0.7 & 2.5 & 200 & 0.702 (0.053) & 2.504 (0.414) \\
    &     & 500 & 0.704 (0.037) & 2.494 (0.222) \\

0.9 & 3.5 & 200 & 0.822 (0.049) & 3.349 (2.288) \\
    &     & 500 & 0.857 (0.042) & 3.352 (1.968) \\

\bottomrule
\end{tabular}
\end{table}

\begin{table}[htbp]
\caption{Maximum likelihood estimates for the SDLB-$p$ distribution.}
\label{tab:4.2}
\centering
\small
\setlength{\tabcolsep}{5pt}
\renewcommand{\arraystretch}{0.85}

\begin{tabular}{ccccccc}
\toprule
\multicolumn{3}{c}{True values} &
\multicolumn{1}{c}{Sample size} &
\multicolumn{3}{c}{MLE estimate and empirical standard deviation} \\
\cmidrule(lr){1-3}
\cmidrule(lr){5-7}
$q_1$ & $q_2$ & $a$ & $n$
& $\hat q_1$ & $\hat q_2$ & $\hat a$ \\
\midrule

\multicolumn{7}{c}{\textit{$p=2,\quad m=5$}} \\
\midrule

0.2 & 0.2 & 0.0 & 200
& 0.170 (0.064) & 0.177 (0.061) & 0.016 (0.210) \\
    &     &     & 500
& 0.183 (0.043) & 0.181 (0.043) & 4.996 (0.128) \\

0.2 & 0.8 & 0.0 & 200
& 0.244 (0.262) & 0.681 (0.218) & 4.990 (0.495) \\
    &     &     & 500
& 0.206 (0.187) & 0.751 (0.100) & 0.045 (0.344) \\

0.8 & 0.2 & 2.5 & 200
& 0.565 (0.289) & 0.416 (0.326) & 2.923 (0.727) \\
    &     &     & 500
& 0.644 (0.219) & 0.399 (0.293) & 2.846 (0.580) \\

0.8 & 0.8 & 2.5 & 200
& 0.512 (0.407) & 0.511 (0.410) & 2.491 (1.469) \\
    &     &     & 500
& 0.533 (0.408) & 0.558 (0.411) & 2.595 (1.294) \\

\addlinespace[4pt]
\multicolumn{7}{c}{\textit{$p=2,\quad m=8$}} \\
\midrule

0.1 & 0.1 & 3.5 & 200
& 0.099 (0.029) & 0.099 (0.029) & 3.498 (0.068) \\
    &     &     & 500
& 0.100 (0.018) & 0.100 (0.018) & 3.501 (0.040) \\

0.1 & 0.5 & 3.5 & 200
& 0.108 (0.079) & 0.494 (0.041) & 3.485 (0.217) \\
    &     &     & 500
& 0.106 (0.052) & 0.497 (0.027) & 3.487 (0.135) \\

0.5 & 0.1 & 6.5 & 200
& 0.495 (0.039) & 0.107 (0.071) & 6.509 (0.194) \\
    &     &     & 500
& 0.499 (0.022) & 0.101 (0.041) & 6.501 (0.097) \\

0.5 & 0.5 & 6.5 & 200
& 0.494 (0.099) & 0.484 (0.102) & 6.483 (0.339) \\
    &     &     & 500
& 0.493 (0.052) & 0.500 (0.048) & 6.511 (0.146) \\

\addlinespace[4pt]
\multicolumn{7}{c}{\textit{$p=4,\quad m=5$}} \\
\midrule

0.2 & 0.2 & 0.0 & 200
& 0.173 (0.064) & 0.173 (0.063) & 0.001 (0.205) \\
    &     &     & 500
& 0.183 (0.040) & 0.183 (0.042) & 5.000 (0.117) \\

0.2 & 0.8 & 0.0 & 200
& 0.234 (0.256) & 0.689 (0.207) & 0.014 (0.504) \\
    &     &     & 500
& 0.194 (0.179) & 0.757 (0.089) & 0.059 (0.342) \\

0.8 & 0.2 & 2.5 & 200
& 0.575 (0.299) & 0.402 (0.332) & 2.885 (0.776) \\
    &     &     & 500
& 0.637 (0.237) & 0.392 (0.300) & 2.854 (0.613) \\

0.8 & 0.8 & 2.5 & 200
& 0.515 (0.405) & 0.505 (0.405) & 2.652 (1.555) \\
    &     &     & 500
& 0.567 (0.405) & 0.527 (0.413) & 2.436 (1.267) \\

\addlinespace[4pt]
\multicolumn{7}{c}{\textit{$p=4,\quad m=8$}} \\
\midrule

0.1 & 0.1 & 3.5 & 200
& 0.098 (0.030) & 0.099 (0.029) & 3.502 (0.064) \\
    &     &     & 500
& 0.099 (0.020) & 0.099 (0.018) & 3.501 (0.044) \\

0.1 & 0.5 & 3.5 & 200
& 0.111 (0.077) & 0.493 (0.043) & 3.474 (0.211) \\
    &     &     & 500
& 0.105 (0.054) & 0.496 (0.026) & 3.486 (0.138) \\

0.5 & 0.1 & 6.5 & 200
& 0.495 (0.041) & 0.106 (0.075) & 6.505 (0.201) \\
    &     &     & 500
& 0.499 (0.024) & 0.103 (0.044) & 6.507 (0.110) \\

0.5 & 0.5 & 6.5 & 200
& 0.484 (0.102) & 0.491 (0.100) & 6.515 (0.339) \\
    &     &     & 500
& 0.497 (0.054) & 0.498 (0.053) & 6.502 (0.163) \\

\bottomrule
\end{tabular}

\end{table}

\begin{table}[!ht]
\caption{Method of moments estimates for the GWGB-$p$ distribution.}
\label{tab:4.3}
\centering
\small
\renewcommand{\arraystretch}{0.85}

\begin{tabular}{ccccc}
\toprule
\multicolumn{2}{c}{True values} &
\multicolumn{1}{c}{Sample size} &
\multicolumn{2}{c}{MoM estimate and standard deviation} \\
\cmidrule(lr){1-2}
\cmidrule(lr){4-5}
$q$ & $a$ & $n$ & $\hat q$ & $\hat a$ \\
\midrule






\multicolumn{5}{c}{\textit{$p=2,\quad m=16$}} \\
\midrule
0.1 & 0.5 & 200 & 0.100 (0.019) & 0.500 (0.035) \\
    &     & 500 & 0.100 (0.013) & 0.501 (0.021) \\

0.3 & 0.5 & 200 & 0.299 (0.027) & 0.498 (0.067) \\
    &     & 500 & 0.300 (0.018) & 0.501 (0.041) \\

0.5 & 2.5 & 200 & 0.498 (0.032) & 2.503 (0.129) \\
    &     & 500 & 0.501 (0.020) & 2.501 (0.082) \\

0.7 & 2.5 & 200 & 0.699 (0.045) & 2.489 (0.399) \\
    &     & 500 & 0.703 (0.029) & 2.495 (0.247) \\

0.9 & 3.5 & 200 & 0.820 (0.051) & 3.556 (2.166) \\
    &     & 500 & 0.852 (0.041) & 3.541 (2.071) \\






\addlinespace[3pt]
\multicolumn{5}{c}{\textit{$p=4,\quad m=32$}} \\
\midrule
0.1 & 0.5 & 200 & 0.100 (0.021) & 0.499 (0.033) \\
    &     & 500 & 0.100 (0.013) & 0.501 (0.020) \\

0.3 & 0.5 & 200 & 0.300 (0.028) & 0.497 (0.067) \\
    &     & 500 & 0.299 (0.018) & 0.500 (0.043) \\

0.5 & 2.5 & 200 & 0.502 (0.033) & 2.495 (0.123) \\
    &     & 500 & 0.500 (0.020) & 2.498 (0.079) \\

0.7 & 2.5 & 200 & 0.699 (0.045) & 2.507 (0.403) \\
    &     & 500 & 0.701 (0.028) & 2.495 (0.229) \\

0.9 & 3.5 & 200 & 0.823 (0.049) & 3.430 (2.232) \\
    &     & 500 & 0.856 (0.041) & 3.405 (1.894) \\

\bottomrule
\end{tabular}
\end{table}

For the GWGB-$p$ distribution, the precision of the estimator of 
$a$ improves as the true value of the concentration parameter $q$ decreases 
from $0.9$ to $0.1$, as shown in Table \ref{tab:4.1}. As $q$ increases, the 
probability mass becomes more dispersed over the sample space, making the 
principal modal direction more difficult to estimate. Conversely, smaller 
values of $q$ produce greater concentration around the modal direction and, 
consequently, more precise estimates of $a$. This behaviour is also consistent 
with the limiting properties of the pmf in \eqref{eq:3.2}. As $q\to0$, the 
distribution degenerates at the direction $2\pi a/\mathrm{m}$ when $a$ is an 
integer. For non-integer $a$, the probability mass concentrates on the category 
nearest to the principal direction, and when the principal direction lies exactly 
halfway between two adjacent categories, the limiting probability mass is 
equally divided between them. In contrast, as $q\to1$, the distribution 
converges to the discrete uniform distribution and therefore loses information 
about the principal modal direction. These limiting cases are discussed in 
Section~\ref{sec:3.1.1}.

For the SDLB-$p$ distribution, the precision of the estimator of $a$ is likewise strongly influenced by the decay parameters $q_1$ and $q_2$. When both parameters are small, for example $q_1=q_2=0.1$, the estimates of $a$ exhibit substantially lower variability than when both parameters are large, such as $q_1=q_2=0.8$. Moreover, an increase in either $q_1$ or $q_2$, while holding the other parameter fixed, generally leads to increased variability in the estimates of $a$, as shown in Tables \ref{tab:4.2} and \ref{tab:4.4}. This behaviour has the same interpretation as for the GWGB-$p$ distribution, smaller values of $q_1$ and $q_2$ produce greater concentration of probability mass around the modal direction, whereas larger values of either parameter lead to a more dispersed distribution and consequently
less precise estimation of the location parameter. This observation is also
consistent with the limiting properties of the pmf in \eqref{eq:4.7},
discussed in Section~\ref{sub:3.2.1}.

Overall, the simulation results show that parameter estimation becomes more
difficult as the distributions become more dispersed. For the GWGB-$p$
distribution, this occurs when $q$ is close to one, whereas for the SDLB-$p$
distribution, similar behaviour is observed when one or both of $q_1$ and
$q_2$ are close to one. In these settings, the estimates of the location
parameter show greater variability, while the estimates of the concentration
or decay parameters exhibit noticeable finite-sample bias. These effects
remain evident at the larger sample size considered.

\begin{table}[htbp]
\caption{Method of moments estimates for the SDLB-$p$ distribution.}
\label{tab:4.4}
\centering
\small
\setlength{\tabcolsep}{5pt}
\renewcommand{\arraystretch}{0.85}

\begin{tabular}{ccccccc}
\toprule
\multicolumn{3}{c}{True values} &
\multicolumn{1}{c}{Sample size} &
\multicolumn{3}{c}{MoM estimate and empirical standard deviation} \\
\cmidrule(lr){1-3}
\cmidrule(lr){5-7}
$q_1$ & $q_2$ & $a$ & $n$
& $\hat q_1$ & $\hat q_2$ & $\hat a$ \\
\midrule

\multicolumn{7}{c}{\textit{$p=2,\quad m=5$}} \\
\midrule

0.2 & 0.2 & 0.0 & 200
& 0.170 (0.057) & 0.174 (0.055) & 0.010 (0.193) \\
    &     &     & 500
& 0.182 (0.039) & 0.180 (0.039) & 4.996 (0.118) \\

0.2 & 0.8 & 0.0 & 200
& 0.254 (0.255) & 0.682 (0.207) & 4.978 (0.478) \\
    &     &     & 500
& 0.220 (0.186) & 0.746 (0.101) & 0.030 (0.342) \\

0.8 & 0.2 & 2.5 & 200
& 0.568 (0.286) & 0.414 (0.324) & 2.906 (0.739) \\
    &     &     & 500
& 0.646 (0.222) & 0.381 (0.299) & 2.804 (0.591) \\

0.8 & 0.8 & 2.5 & 200
& 0.518 (0.407) & 0.500 (0.410) & 2.449 (1.486) \\
    &     &     & 500
& 0.513 (0.417) & 0.547 (0.417) & 2.585 (1.263) \\

\addlinespace[4pt]
\multicolumn{7}{c}{\textit{$p=2,\quad m=8$}} \\
\midrule

0.1 & 0.1 & 3.5 & 200
& 0.101 (0.033) & 0.099 (0.030) & 3.496 (0.076) \\
    &     &     & 500
& 0.100 (0.023) & 0.101 (0.024) & 3.502 (0.061) \\

0.1 & 0.5 & 3.5 & 200
& 0.118 (0.103) & 0.492 (0.046) & 3.482 (0.303) \\
    &     &     & 500
& 0.107 (0.071) & 0.497 (0.029) & 3.497 (0.197) \\

0.5 & 0.1 & 6.5 & 200
& 0.492 (0.045) & 0.115 (0.100) & 6.506 (0.291) \\
    &     &     & 500
& 0.498 (0.028) & 0.105 (0.070) & 6.498 (0.191) \\

0.5 & 0.5 & 6.5 & 200
& 0.495 (0.096) & 0.488 (0.095) & 6.487 (0.305) \\
    &     &     & 500
& 0.495 (0.054) & 0.499 (0.053) & 6.505 (0.163) \\

\addlinespace[4pt]
\multicolumn{7}{c}{\textit{$p=4,\quad m=5$}} \\
\midrule

0.2 & 0.2 & 0.0 & 200
& 0.172 (0.057) & 0.171 (0.058) & 4.998 (0.192) \\
    &     &     & 500
& 0.182 (0.036) & 0.182 (0.038) & 4.999 (0.107) \\

0.2 & 0.8 & 0.0 & 200
& 0.247 (0.251) & 0.687 (0.200) & 4.992 (0.485) \\
    &     &     & 500
& 0.206 (0.177) & 0.753 (0.085) & 0.046 (0.338) \\

0.8 & 0.2 & 2.5 & 200
& 0.577 (0.299) & 0.396 (0.334) & 2.863 (0.786) \\
    &     &     & 500
& 0.642 (0.235) & 0.377 (0.303) & 2.808 (0.618) \\

0.8 & 0.8 & 2.5 & 200
& 0.518 (0.404) & 0.499 (0.405) & 2.590 (1.541) \\
    &     &     & 500
& 0.554 (0.410) & 0.513 (0.417) & 2.420 (1.249) \\

\addlinespace[4pt]
\multicolumn{7}{c}{\textit{$p=4,\quad m=8$}} \\
\midrule

0.1 & 0.1 & 3.5 & 200
& 0.099 (0.035) & 0.100 (0.033) & 3.504 (0.084) \\
    &     &     & 500
& 0.098 (0.022) & 0.099 (0.022) & 3.502 (0.064) \\

0.1 & 0.5 & 3.5 & 200
& 0.123 (0.102) & 0.491 (0.045) & 3.464 (0.299) \\
    &     &     & 500
& 0.103 (0.074) & 0.497 (0.027) & 3.508 (0.208) \\

0.5 & 0.1 & 6.5 & 200
& 0.492 (0.047) & 0.121 (0.104) & 6.526 (0.300) \\
    &     &     & 500
& 0.498 (0.028) & 0.107 (0.073) & 6.503 (0.205) \\

0.5 & 0.5 & 6.5 & 200
& 0.491 (0.092) & 0.493 (0.092) & 6.507 (0.296) \\
    &     &     & 500
& 0.498 (0.055) & 0.499 (0.055) & 6.503 (0.169) \\

\bottomrule
\end{tabular}

\end{table}


\section{Data Analysis}\label{sec:6}
In this section, we apply the proposed models to two discrete angular datasets exhibiting $p$--fold rotational symmetry.

\textbf{Dataset 1:}
An example of axial data ($p=2$) from geology is presented in Table \ref{tab:1.2}, which consists of the orientations of the shortest projection elongations of sand grains from Gulf Coast beach sand.

\begin{table}[!ht]
\centering
\caption{Orientations of the shortest projection elongations
of sand grains \citep{curray1956}.}
\label{tab:1.2}

\small
\setlength{\tabcolsep}{2pt}

\begin{tabular}{@{}lccccccccc@{}}
\toprule
\textbf{Angle ($^\circ$)}
& \textbf{0--9}
& \textbf{10--19}
& \textbf{20--29}
& \textbf{30--39}
& \textbf{40--49}
& \textbf{50--59}
& \textbf{60--69}
& \textbf{70--79}
& \textbf{80--89} \\
\midrule
\textbf{Frequency}
& 244 & 262 & 246 & 290 & 284
& 314 & 326 & 340 & 371 \\
\midrule
\textbf{Angle ($^\circ$)}
& \textbf{90--99}
& \textbf{100--109}
& \textbf{110--119}
& \textbf{120--129}
& \textbf{130--139}
& \textbf{140--149}
& \textbf{150--159}
& \textbf{160--169}
& \textbf{170--179} \\
\midrule
\textbf{Frequency}
& 401 & 382 & 332 & 322 & 295
& 230 & 256 & 263 & 281 \\
\bottomrule
\end{tabular}
\end{table}

The dataset is grouped into $18$ discrete bins, each of width $10^\circ$. The tabulated frequencies are taken from Chapter 1 of \cite{mardia2009}, who sourced the original data from \citep{curray1956}. For the proposed discrete models, the categories are indexed by the integers $k=0,1,\ldots,17$ and represented by the angular points $0^\circ,10^\circ,\ldots,170^\circ$.

We fit the GWGB-$2$ and SDLB-$2$ distributions to the data using MLE. The dataset consists of $18$ angular categories. For the GWGB-$2$ distribution, the number of categories is $\mathrm{h}=\mathrm{m}/p$; thus, with $h=18$ and $p=2$, we set $\mathrm{m}=\mathrm{h}p=36$. For the SDLB-$2$ distribution, $\mathrm{m}$ represents the number of angular categories, and hence we set $\mathrm{m}=18$. The resulting maximum likelihood estimates, together with their standard errors, negative log-likelihood (NLL), and Akaike information criterion (AIC) values, are presented in Table~\ref{tab:5.2}.

\begin{table}[!ht]
\centering
\small
\caption{Estimates (with standard error in parentheses) for Dataset 1} 
\label{tab:5.2}
\begin{tabular}{ m{2.5cm} | m{4cm} | m{4cm}}  
\hline
 $\textbf{Parameters}$ & \textbf{Estimates of GWGB-$2$} & \textbf{Estimates of SDLB-$2$}\\
 \hline
 $\hat{q}$ &  0.889 (0.005)  & \\
 $\hat{q_1}$ & & 0.816  (0.049)   \\
 $\hat{q_2}$ & & 0.929  (0.016)   \\
 $\hat{a}$ & 8.762 (0.216) & 9.443  (0.393) \\ 
 NLL & 15660.200 & 15658.800\\ 
 AIC & 31324.400 & 31323.600   \\
 \hline
\end{tabular}
\end{table}

The SDLB-$2$ distribution has a slightly smaller AIC than the GWGB-$2$ distribution. However, the AIC values of the two models are close, with a difference of only $0.8$. Therefore, we consider both distributions for interpreting the orientation pattern.

For the GWGB-$2$ distribution, the estimated location is $\hat{a}=8.762$, corresponding to the principal orientation axis
\[
\frac{2\pi\hat{a}_{MLE}}{\mathrm{m}}
=
\frac{2\pi(8.762)}{36}
\approx 87.62^\circ.
\]
Thus, the fitted GWGB-$2$ distribution indicates a principal orientation axis of approximately $87.62^\circ$. The estimate $\hat q=0.889$ controls the concentration of the probability mass around the principal orientation. As shown in Figure~\ref{fig:5.1}, the fitted GWGB-$2$ distribution exhibits a relatively broad spread around this orientation.

For the SDLB-$2$ distribution, the estimated location is $\hat{a}=9.443$, corresponding to the principal orientation axis
\[
\frac{2\pi\hat{a}_{MLE}}{p\mathrm{m}}
=
\frac{2\pi(9.443)}{2(18)}
\approx 94.43^\circ.
\]
Thus, the fitted SDLB-$2$ distribution indicates a principal orientation axis of approximately $94.43^\circ$.

The estimates of the decay parameters are $\hat{q}_1=0.816$ and $\hat{q}_2=0.929$, which determine the rates at which the pmf decreases in the increasing- and decreasing-index directions, respectively, from the principal axis $\left(2\pi\hat{a}/p\mathrm{m}\right)$. As shown in Figure~\ref{fig:5.1}, since $\hat{q}_2>\hat{q}_1$, the fitted SDLB-$2$ pmf exhibits skewness around the principal axis, with probabilities decreasing more gradually from $90^\circ$ toward $0^\circ$ than from $90^\circ$ toward $170^\circ$.

\begin{figure}[!ht]
    \centering
    \includegraphics[width=\linewidth]{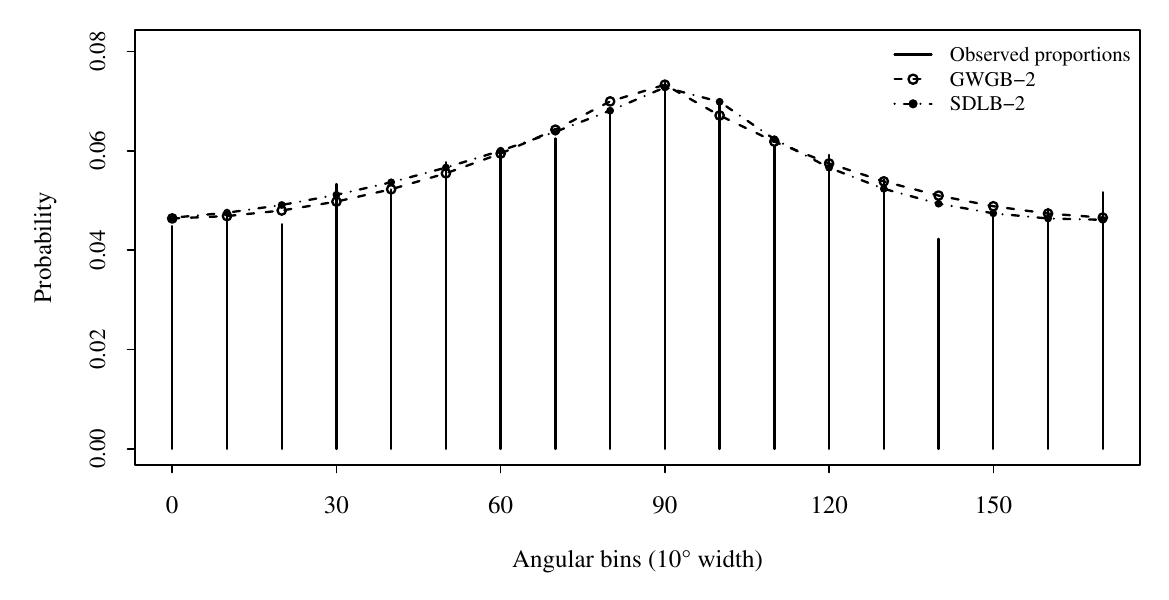}
    \caption{Observed proportions and fitted probabilities for Dataset 1 under the fitted GWGB-$2$ and SDLB-$2$ distributions.}
    \label{fig:5.1}
\end{figure}

Figure~\ref{fig:5.1} shows the observed proportions together with the fitted probabilities under both the GWGB-$2$ and SDLB-$2$ distributions. Both fitted distributions capture the broad preferred orientation of the sand-grain data, while the SDLB-$2$ distribution additionally accommodates different rates of decrease on the two sides of its principal orientation through the parameters $q_1$ and $q_2$.

\noindent
Next, we illustrate the application of our proposed models for the case $p = 4$.

\noindent
\textbf{Dataset 2:}
In urban morphology, Manhattan is well known for its highly regular orthogonal street layout, and we obtained its street network from the \textit{OpenStreetMap} (OSM) database using the \texttt{osmdata} package in \textsf{R}. The street geometries within the administrative boundary of Manhattan were extracted and decomposed into straight line segments. The compass bearing of each segment was then computed. Following \cite{boeing2019}, the street orientations were grouped into equally spaced angular bins of width $10^\circ$. The Manhattan street network is predominantly organized along two perpendicular directions, producing an orientation distribution concentrated near two principal axes approximately $90^\circ$ apart, consistent with the findings of \cite{boeing2019}. This pattern is also evident in Figure~\ref{fig:5.2}, which represents the street network and corresponding histogram.

\begin{figure}[!ht]
    \centering
    \includegraphics[width=\linewidth]{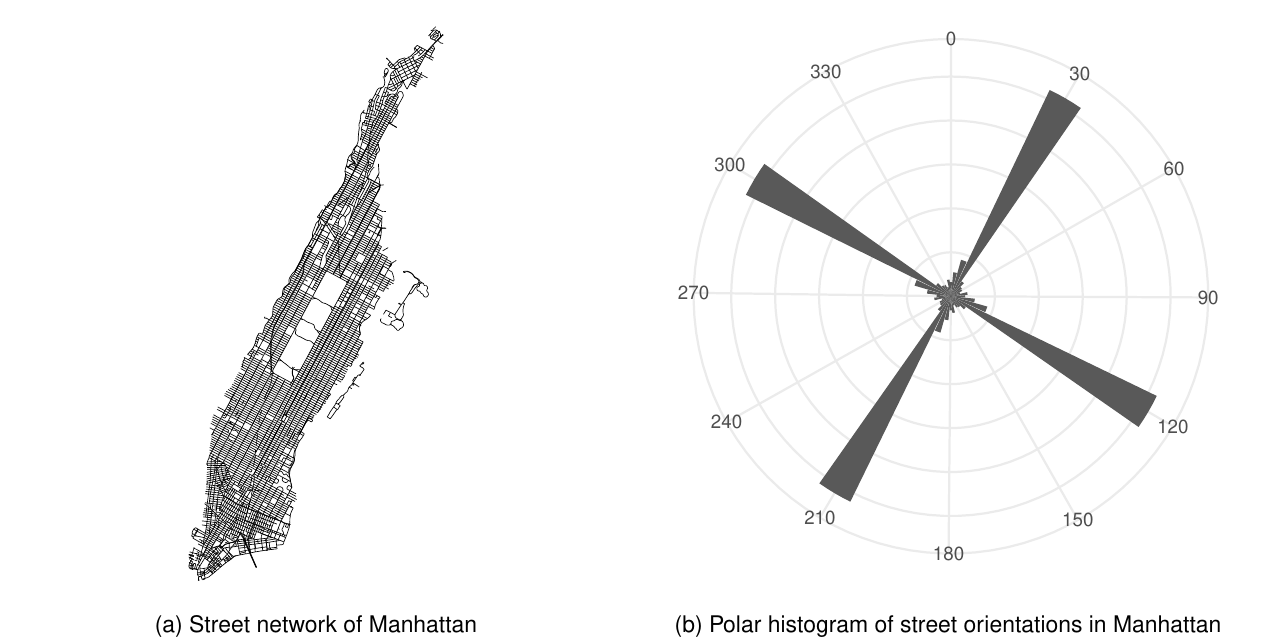}
    \caption{Street network and corresponding histogram for Manhattan}
    \label{fig:5.2}
\end{figure}

Since the Manhattan street network exhibits approximately four-fold rotational symmetry ($p=4$), the street orientations were folded modulo $90^\circ$ and grouped into nine equal-width bins of $10^\circ$. The resulting frequency counts are presented in Table~\ref{tab:1.3}.

\begin{table}[!ht]
\centering
\caption{Frequency distribution of Manhattan street orientations after reduction modulo $90^\circ$.}
\label{tab:1.3}
\begin{tabular}{|l|c|c|c|c|c|c|c|c|c|}
\hline
\textbf{Bin center ($^\circ$)} & \textbf{0} & \textbf{10} & \textbf{20} & \textbf{30} & \textbf{40} & \textbf{50} & \textbf{60} & \textbf{70} & \textbf{80} \\ \hline
\textbf{Frequency} & 1298 & 2195 & 3449 & 20819 & 1676 & 1212 & 1130 & 870 & 1526 \\ \hline
\end{tabular}%
\end{table}

Now we fit the GWGB-$4$ and SDLB-$4$ distributions presented in Sections~\ref{sub:3.1} and~\ref{sub:3.2}, respectively. To fit these models, the nine orientation categories were indexed by the integers $0,1,\ldots,8$, corresponding to the angular bin centers $0^\circ,10^\circ,\ldots,80^\circ$. We set $\mathrm{m}=36$ for GWGB-$4$ and $\mathrm{m}=9$ for SDLB-$4$. The results are presented in Table~\ref{tab:1.4}.

\begin{table}[!ht]
\centering
\small
\caption{Estimates (with standard error in parentheses) for Dataset 2} 
\label{tab:1.4}
\begin{tabular}{ m{2.5cm} | m{4cm} | m{4cm}}  
\hline
 $\textbf{Parameters}$ & \textbf{Estimates of GWGB-$4$} & \textbf{Estimates of SDLB-$4$}\\
 \hline
 $\hat{q}$ &  0.379 (0.002)  & \\
 $\hat{q_1}$ & & 0.237  (0.004)   \\
 $\hat{q_2}$ & & 0.467  (0.003)   \\
 $\hat{a}$ & 2.999 (0.001) & 3.000  (0.001) \\ 
 NLL & 54274.280 & 53475.550\\ 
 AIC & 108552.600 & 106957.100\\
 \hline
\end{tabular}
\end{table}

The estimates of the parameter $a$, which specifies the location of the highest probability mass, are similar for both models. However, since SDLB-$4$ has a smaller AIC than GWGB-$4$, we choose the second model for interpretation. The estimate $\hat{a} = 3.000$ corresponds to the principal direction
\[\frac{2 \pi \hat{a}}{p\mathrm{m}} = \frac{2\pi (3.000)}{4(9)} \approx 30^\circ.\]
Thus, the estimated principal street-grid orientation is approximately $30^\circ$ clockwise from the north ($0^\circ$). Under four-fold rotational symmetry, this principal orientation is equivalent to the dominant directions $30^\circ, 120^\circ, 210^\circ$, and $300^\circ$
in the original $360^\circ$ angular space, and these four dominant directions are clearly visible in the polar histogram of street orientations (Figure~\ref{fig:5.2}) and reflect the characteristic orthogonal organization of Manhattan’s street network.

\begin{figure}[!ht]
    \centering
    \includegraphics[width=\linewidth]{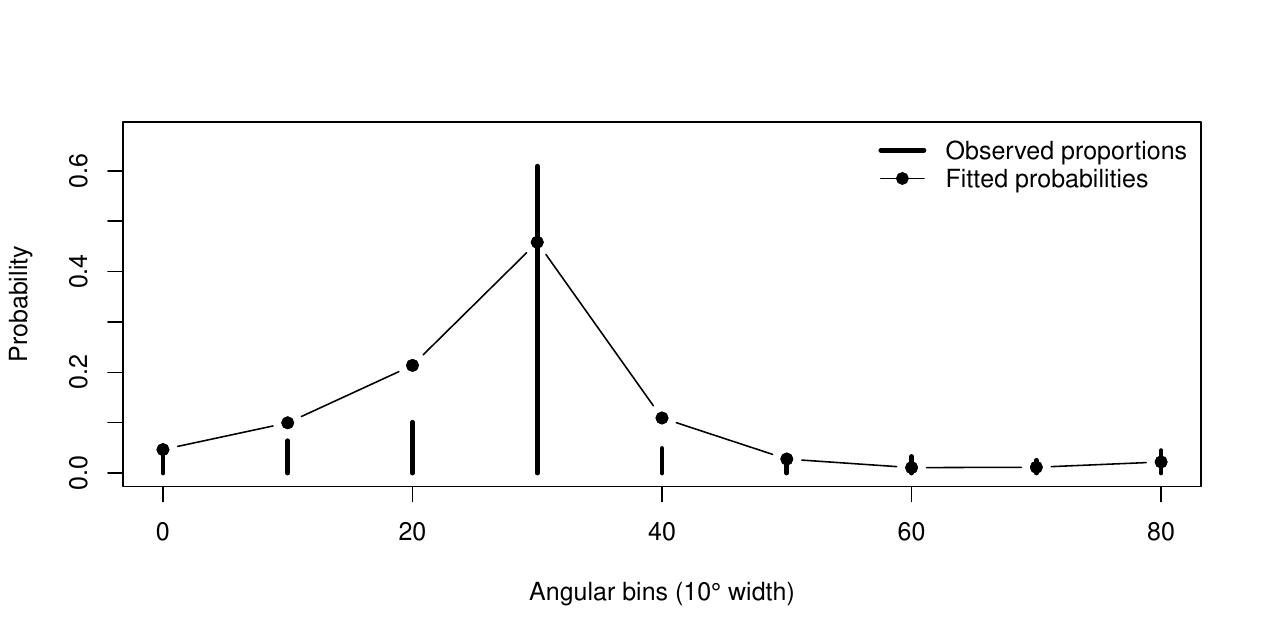}
    \caption{Observed proportions and fitted probabilities for Dataset 2 under SDLB-$4$}
    \label{fig:5.3}
\end{figure}

The estimates $\hat{q}_{1}=0.237$ and $\hat{q}_{2}=0.467$ determine the
rates at which the pmf decreases in the increasing- and decreasing-index
directions, respectively, from the principal direction
$\left(2\pi\hat{a}/p\mathrm{m}\right)$. As shown in Figure~\ref{fig:5.3}, since $\hat{q}_{1}<\hat{q}_{2}$, the fitted pmf is skewed around the principal direction, with probabilities decreasing more rapidly from $30^\circ$ toward $40^\circ$ and more gradually from $30^\circ$ toward $20^\circ$.

\subsection{Comparison of proposed models with alternatives}
We compare the proposed GWGB-$2$ and SDLB-$2$ distributions for Dataset 1 with discretized versions of the continuous axial distributions: the sine-skewed axial Jones–Pewsey (SSAJP), symmetric axial Jones–Pewsey (AJP), sine-skewed axial von Mises (SSAVM), and sine-skewed axial wrapped Cauchy (SSAWC) distributions \citep{abe2012}. To ensure that these continuous distributions were compared on the same finite angular sample space as the proposed distributions, they were discretized by assigning to each observed angular category the probability obtained by numerically integrating the corresponding continuous density over the interval represented by that category. The resulting discretized distributions are denoted by D--SSAJP, D--AJP, D--SSAVM, and D--SSAWC, respectively. The parameter estimates, together with the corresponding AIC values, are presented in Table~\ref{tab:1.5}.

\begin{table}[!ht]
\centering
\small
\caption{Maximum likelihood estimates and model-selection criteria for the
models fitted to Dataset 1}
\label{tab:1.5}
\begin{tabular}{m{2cm} | m{7.5cm} | m{1.5cm}}
\hline
\textbf{Model}
& \textbf{Maximum likelihood estimates}
& \textbf{AIC} \\
\hline

SDLB-$2$
& $\hat{q}_{1}= 0.816,\quad
   \hat{q}_{2}= 0.929,\quad
   a^*= 94.430^\circ$
& 31323.600 \\

GWGB-$2$
& $\hat{q}= 0.889,\quad
   a^* = 2\pi\hat{a}/36 = 87.620^\circ$
& 31324.400 \\

D-SSAJP
& $\hat{\kappa}= 7.265,\quad
   \hat{\psi} =  8.169,\quad
   \hat{\lambda}= -0.183,\quad
   \hat{\mu}= 141.692 ^{\circ}$
& 31324.430 \\

D-AJP
& $\hat{\kappa}= 0.512,\quad
   \hat{\psi} =  -3.343,\quad
   \hat{\mu}= 87.605^{\circ}$
& 31326.280 \\

D--SSAWC
& $\hat{\rho}= 0.319,\quad
   \hat{\lambda}= -0.050,\quad
   \hat{\mu}= 92.132^{\circ}$
& 31331.170  \\

D--SSAVM
& $\hat{\kappa}= 0.786,\quad
   \hat{\lambda}= -0.150,\quad
   \hat{\mu}= 108.633^{\circ}$
& 31334.150 \\

\hline
\end{tabular}
\end{table}

Among the models considered, the proposed SDLB-$2$ distribution attained the smallest
AIC value of $31323.600$. However, the AIC differences between the SDLB-$2$
and the GWGB-$2$ and D-SSAJP distributions are only $0.800$ and $0.830$,
respectively. Thus, these three models provide essentially comparable empirical
support for Dataset 1. The D-AJP, D--SSAWC, and D--SSAVM distributions have
larger AIC values. An additional advantage of the proposed GWGB-$2$ and
SDLB-$2$ distributions is that their category probabilities are defined
directly on the observed discrete angular support, whereas the discretized
continuous alternatives require numerical integration over the corresponding
angular intervals during likelihood evaluation.

\section{Conclusion}\label{sec:7}
In this article, we proposed a unified framework for constructing native discrete distributions on the reduced angular space $\Omega_\Theta \subset [0, 2\pi/p)$ for modelling directional observations exhibiting $p$--fold rotational symmetry. We discussed particular distributions arising from this framework, together with some of their properties and parameter estimation methods. We further demonstrated the practical applicability of the proposed distributions through two real datasets exhibiting $2$--fold and $4$--fold rotational symmetry.

To the best of our knowledge, existing continuous models have focused predominantly on the axial case ($p=2$). Consequently, for Dataset 1, we compared the performance of the proposed axial GWGB-$2$ and SDLB-$2$ distributions with the discretized versions of well-suited continuous axial distributions.

The proposed framework currently applies to univariate discrete directional observations. Future work may extend it to regression models for discrete directional responses with p--fold rotational symmetry and to bivariate discrete directional distributions whose components are defined on the reduced angular spaces corresponding to $p_1$--fold and $p_2$--fold rotational symmetry.
 
\section*{Declarations}
\noindent
\textbf{Funding}\\
This project is partially funded by the Government of India as part of the Start-up Research Grant. The funding is provided through the Science and Engineering Research Board, Department of Science and Technology, India to Dr. Jayant Jha (SRG/2022/000151). Brajesh Kumar Dhakad is supported by the National Board for Higher Mathematics (NBHM) (0203/4/2022/R\&D-II/ 3262). The funders had no role in the methodological development, data analysis, decision to publish, or preparation of the manuscript.

\noindent
\textbf{Competing Interests} \\
The authors have no relevant financial or non-financial interests to disclose.

\noindent
\textbf{Data Availability} \\
 Dataset 1 and Dataset 2 are obtained from the sources that are detailed in the section~\ref{sec:6}.

\noindent
\textbf{Code Availability} \\
The R codes used for the data analysis are available upon request.

\bibliography{sn-bibliography}


\end{document}